\documentclass[11pt,a4paper,reqno,intlimits]{amsart}

\usepackage{amssymb}
\usepackage{chicago}
\usepackage{color}
\usepackage[mathscr]{euscript}
\usepackage{xspace}
\usepackage{epsfig,rotating}

\begin{document}

%% ================================================================ %%
\frenchspacing

\theoremstyle{plain}
\newtheorem{theorem}{Theorem}[section]
\newtheorem{lemma}[theorem]{Lemma}
\newtheorem{proposition}[theorem]{Proposition}
\newtheorem{claim}[theorem]{Claim}
\newtheorem{corollary}[theorem]{Corollary}

\theoremstyle{definition}
\newtheorem{remark}[theorem]{Remark}
\newtheorem{note}{Note}[section]
\newtheorem{definition}[theorem]{Definition}
\newtheorem{example}[theorem]{Example}
\newtheorem*{ackn}{Acknowledgements}
\newtheorem{assumption}{Assumption}
\newtheorem*{assuAC}{Assumption ($\mathbb{AC}$)}
\newtheorem*{assuEM}{Assumption ($\mathbb{EM}$)}
%% ================================================================ %%
\renewcommand{\theequation}{\thesection.\arabic{equation}}
\numberwithin{equation}{section}

\newcommand{\Law}{\ensuremath{\mathop{\mathrm{Law}}}}
\newcommand{\loc}{{\mathrm{loc}}}

\newcommand{\prozess}[1][L]{{\ensuremath{#1=(#1_t)_{0\le t\le T}}}\xspace}
\newcommand{\prazess}[1][L]{{\ensuremath{#1=(#1_t)_{0\le t\le T^*}}}\xspace}
% ================================================================ %%
\def\P{\ensuremath{\mathrm{I\kern-.2em P}}}
\def\E{\mathrm{I\kern-.2em E}}
\def\bF{\mathbf{F}}
\def\F{\ensuremath{\mathcal{F}}}
\def\R{\ensuremath{\mathbb{R}}}
\def\C{\ensuremath{\mathbb{C}}}
\def\fF{\ensuremath{\mathrm{I\kern-.2em F}}}

\def\LtT{\ensuremath{(L_t)_{t\in[0,T^*]}}}
\def\lev{L\'{e}vy\xspace}
\def\lk{L\'{e}vy--Khintchine\xspace}
\def\smmg{semimartingale\xspace}
\def\mg{martingale\xspace}
\def\tih{time-inhomogeneous\xspace}
\def\chartri{\ensuremath{(b,\sigma,\nu)}}

\def\ud{\ensuremath{\mathrm{d}}}
\def\dsdx{\ensuremath{(\ud s, \ud x)}}
\def\dtdx{\ensuremath{(\ud t, \ud x)}}
\def\intrr{\ensuremath{\int_{\R}}}
\def\MM{\ensuremath{\mathscr{M}}}

\def\Rp{\mathbb{R}_{+}}

\def\EM{\ensuremath{(\mathbb{EM})}\xspace}
\def\ES{\ensuremath{(\mathbb{ES})}\xspace}
\def\AC{\ensuremath{(\mathbb{AC})}\xspace}

\def\e{\mathrm{e}}
\def\intrr{\int_{\R}}
\def\dt{\ud t}
\def\ds{\ud s}
\def\dx{\ud x}
\def\dy{\ud y}
%% ------------------ Definitions: Forward Libor ------------------ %%
\def\libor{L(t,T)}
\def\fttj{F(t,T^*_{j},T^*_{j-1})}
\def\fTtj{F(T^*_j,T^*_{j},T^*_{j-1})}
\def\Ptj{\P_{T^*_{j-1}}}
\def\fttjp{\widetilde{F}(t,T_{j},T_{j+1})}
\def\tb{\bar{T}}
\def\tj{T^*_j}
\def\tjp{T^*_{j-1}}
\def\half{\frac{1}{2}}
\def\thstj{\eta(s,T^*_j)}
\def\LibT{L(t,T_j^*)}
\def\MeaT{\P_{T^*_{j-1}}}
\def\volT{\eta(s,T_j^*)}
\def\vol2T{(\eta(s,T_j^*))^2}
\def\LevT{L_s^{T^*_{j-1}}}
%% ---------------------------------------------------------------- %%
%% ================================================================ %%

\title[Compositions in L\'evy term structure models]
      {On the valuation of compositions in L\'evy term structure models}

\author{Wolfgang Kluge}
\author{Antonis Papapantoleon}

\address{BNP Paribas, 10 Harewood Avenue, London NW1 6AA, United Kingdom}
\email{wolfgang.kluge@uk.bnpparibas.com}

\address{Financial and Actuarial Mathematics, Vienna University of Technology,
         Wiedner Hauptstrasse 8/105, 1040 Vienna, Austria }
\email{papapan@fam.tuwien.ac.at}


\thanks{We thank Ernst Eberlein for motivating discussions and helpful comments.
        Both authors gratefully acknowledge the financial support from the Deutsche
        Forschungsgemeinschaft (DFG, EB 66/9-2). A.P. gratefully acknowledges the
        financial support from the Austrian Science Fund (FWF grant Y328, START Prize)}

\date{}
\maketitle
\pagestyle{myheadings}

\begin{abstract}
We derive explicit valuation formulae for an exotic path-dependent interest
rate derivative, namely an option on the composition of LIBOR rates. The
formulae are based on Fourier transform methods for option pricing. We
consider two models for the evolution of interest rates: an HJM-type forward
rate model and a LIBOR-type forward price model. Both models are driven
by a \tih \lev process.
\end{abstract}

\section{Introduction}
\label{intro}

The main aim of this paper is to derive simple and analytically tractable
valuation formulae for an exotic path dependent interest rate derivative,
namely an option on the composition of LIBOR rates. The formulae make use
of Fourier transform techniques, see e.g. Eberlein, Glau, and Papapantoleon
\citeyear{EberleinGlauPapapantoleon08},
and the change-of-numeraire technique. There are two models for the term
structure of interest rates considered in this paper: a
Heath--Jarrow--Morton-type forward rate model and a LIBOR-type forward
price model, both driven by a general \tih \lev process.

A standard approach to modeling the term structure of interest rates
is that of \citeN{HeathJarrowMorton92}. In the Heath--Jarrow--Morton
(henceforth HJM) framework subject to modeling are instantaneous
continuously compounded forward rates
which are driven by a $d$-dimensional Wiener process. However, data
from bond markets do not support the use of the normal distribution.
Empirical evidence for the non-Gaussianity of daily returns from
bond market data can be found in \citeN[chapter 5]{Raible00}; the
fit of the normal inverse Gaussian distribution to the same data is
particularly good, supporting the use of \lev processes for modeling
interest rates. Similar evidence appears in the risk-neutral world,
i.e. from caplet implied volatility smiles and surfaces; see
\citeN{EberleinKluge04}.

The \lev forward rate model was developed in
\citeN{EberleinRaible99} and extended to time-inhomogeneous \lev
processes in Eberlein, Jacod, and Raible
\citeyear{EberleinJacodRaible05}. In these models, forward
rates are driven by a (\tih) \lev process; therefore,
the model allows to accurately capture the empirical dynamics of
interest rates, while it is still analytically tractable, so that
closed form valuation formulae for liquid derivatives can be
derived. Valuation formulae for caps, floors, swaptions and range
notes have been derived in \citeANP{EberleinKluge04}
\citeyear{EberleinKluge04,EberleinKluge05}, while estimation and
calibration methods are discussed in \citeANP{EberleinKluge04}
\citeyear{EberleinKluge04,EberleinKluge06}.

Moreover, \citeN{EberleinJacodRaible05} provide a complete classification
of all equivalent martingale measures in the \lev forward rate model.
They also prove that in certain situations -- essentially, if the
driving process is $1$-dimensional -- the set of equivalent
martingale measures becomes a singleton.

The main pitfall of the HJM framework is the assumption of
continuously compounded rates, while in real markets interest
accrues according to a discrete grid, the tenor structure. LIBOR
market models, that is, arbitrage-free term structure models on a
discrete tenor, were constructed in a series of articles by
\shortciteN{SandmannSondermannMiltersen95},
\shortciteN{MiltersenSandmannSondermann97},
\shortciteN{BraceGatarekMusiela97}, and \citeN{Jamshidian97}.
In addition, LIBOR
market models are consistent with the market practice of pricing
caps and floors using Black's formula (cf. \citeNP{Black76}).

Nevertheless, a familiar phenomenon appears: since the model is
driven by a Brownian motion, it cannot be calibrated accurately to
the whole term structure of volatility smiles. As a remedy,
\citeN{EberleinOezkan05} developed a LIBOR model driven by time
inhomogeneous \lev processes. Valuation methods for caps and floors,
using approximation arguments, were presented in \citeN{EberleinOezkan05}
and \citeN{Kluge05}, while calibration issues for this model are
discussed in \citeN{EberleinKluge06}.

The \lev forward price model is a market model based on the \emph{forward
price} -- rather than the \emph{LIBOR rate} -- and driven by time
inhomogeneous \lev processes; it was put forward by \citeANP{EberleinOezkan05}
(2005, pp. 342--343). A detailed construction of the model is
presented in \citeANP{Kluge05} (2005, Chapter 3); there, it is also shown how
this model can be embedded in the \lev forward rate model.

Although the forward LIBOR rate and the forward price differ only by an
additive and a multiplicative constant, the
two specifications lead to models with very different qualitative
and quantitative behavior. In the LIBOR model, LIBOR rates change by
an amount relative to their current level, while in the forward
price model changes do not depend on the actual level (cf.
\citeNP[p. 60]{Kluge05}). There are authors who claim that models
based on the forward process -- also coined ``arithmetic'' or
``Bachelier'' LIBOR models -- are able to better describe the
dynamics of the market than (log-normal) LIBOR market models; see
\citeN{Henrard05}.

Another advantage of the forward price model is that the driving
process remains a time-inhomogeneous \lev process under each forward
measure, hence this model  is parti\-cularly suitable for practical
implementation. The downside is that negative LIBOR rates can occur,
like in an HJM model.

This paper is organized as follows: in section \ref{PIIAC} we review
some basic properties of the driving \tih \lev processes and in section
\ref{LTSMs} we describe the forward rate and forward price frameworks
for modeling the term structure of interest rates. In section \ref{compo}
the payoff of the option on the composition is described and valuation
formulae are derived in the two modeling frameworks. Finally, section
\ref{conclusion} concludes.

\section{Time-inhomogeneous L\'evy processes}
\label{PIIAC}

Let ($\Omega, \F, \fF, \P$) be a complete stochastic basis, where
$\F=\F_{T^*}$ and the filtration $\fF=(\F_t)_{t\in[0,T^*]}$
satisfies the usual conditions; we assume that $T^*\in\Rp$ is a
finite time horizon. The driving process $L=\LtT$ is a \emph{\tih
\lev process}, or a \emph{process} with \emph{independent
increments} and \emph{absolutely continuous} characteristics, in the
sequel abbreviated PIIAC. Therefore, $L$ is an adapted, c\`{a}dl\`{a}g, real-valued
stochastic process with independent increments, starting from zero,
where the law of $L_t$, $t\in[0,T^*]$, is described by the
characteristic function
\begin{align}
\E\left[\e^{iuL_{t}}\right]
 = \exp\int_{0}^{t}\bigg( ib_su - \frac{c_s}{2}u^{2}
  + \int_{\R}(\e^{iux}-1-iux)\lambda_s(\ud x)\bigg)\ud s,
\end{align}
where $b_t\in\R$, $c_t\in\Rp$ and $\lambda_t$ is a \lev measure,
i.e. it satisfies $\lambda_t(\{0\})=0$ and
$\int_{\R}(1\wedge|x|^2)\lambda_t(\ud x)<\infty$, for all
$t\in[0,T^*]$. In addition, the process $L$ satisfies Assumptions
($\mathbb{AC}$) and ($\mathbb{EM}$) given below.

\begin{assuAC}
The triplets ($b_t,c_t,\lambda_t$) satisfy
\begin{eqnarray}
\int_{0}^{T^*}\bigg( |b_t| + c_t
  + \int_{\R}(1\wedge|x|^2)\lambda_t(\ud x) \bigg)\ud t <\infty.
\end{eqnarray}
\end{assuAC}

\begin{assuEM}
There exist constants $M, \varepsilon>0$ such that for every
$u\in[-(1+\varepsilon)M,(1+\varepsilon)M]$
\begin{equation}\label{eq:Int}
    \int_0^{T^*}\int_{\{|x|>1\}}\exp(ux) \lambda_t(\ud x)\ud t<\infty.
\end{equation}
Moreover, without loss of generality, we assume that
$\int_{\{|x|>1\}}\e^{ux}\lambda_t(\ud x)<\infty$ for all
$t\in[0,T^*]$ and all $u\in[-(1+\varepsilon)M,(1+\varepsilon)M]$.
\end{assuEM}

These assumptions render the process \prazess a \emph{special}
semimartingale, therefore it has the canonical decomposition (cf.
Jacod and Shiryaev
\citeyearNP[II.2.38]{JacodShiryaev03}, and
\shortciteNP{EberleinJacodRaible05})
\begin{align}\label{canonical}
 L_{t} = \int_{0}^{t}b_s\ud s
       + \int_{0}^{t}\sqrt{c_s} \ud W_{s}
       + \int_{0}^{t}\int_{\R} x(\mu^{L}-\nu)(\ud s,\ud x),
\end{align}
where $\mu^L$ is the random measure of jumps of the process $L$ and
\prazess[W] is a $\P$-standard Brownian motion. The \emph{triplet of
predictable} or \emph{semimartingale characteristics} of $L$ with
respect to the measure $P$, $\mathbb T(L|P)=(B,C,\nu)$, is
\begin{eqnarray}\label{ch4:triplet}
 B_{t}=\int_{0}^{t}b_s\ud s, &&
 C_t=\int_0^t c_s\ud s,\quad
 \nu([0,t] \times A)=\int_0^t\int_A \lambda_s(\ud x)\ud s,
\end{eqnarray}
where $A\in\mathcal{B}(\R)$. The triplet ($b,c,\lambda$) represents
the \emph{local} or \emph{differential characteristics} of $L$. In
addition, the triplet of semimartingale characteristics ($B,C,\nu$)
determines the distribution of $L$.

We denote by $\theta_s$ the \emph{cumulant generating function} (i.e.
the logarithm of the moment generating function) associated with the
infinitely divisible distribution with \lev triplet
($b_s,c_s,\lambda_s$), i.e. for
$z\in[-(1+\varepsilon)M,(1+\varepsilon)M]$
\begin{align}\label{ch4:cumulant}
\theta_s(z) := b_sz+\frac{c_s}{2}z^{2}
             + \int_{\R}(\e^{zx}-1-zx)\lambda_s(\ud x).
\end{align}
Subject to Assumption \EM, $\theta_s$ is well defined and can be extended
to the complex domain $\C$, for
$z\in\C$ with $\Re z\in[-(1+\varepsilon)M,(1+\varepsilon)M]$ and
the characteristic function of $L_t$ can be written as
\begin{align}
\E\left[\e^{iuL_{t}}\right] = \exp\int_{0}^{t} \theta_s(iu)\ud s.
\end{align}
If $L$ is a (time-homogeneous) \lev process, then
($b_s,c_s,\lambda_s$) and thus also $\theta_s$ do not depend on $s$,
and $\theta$ equals the cumulant generating function of $L_1$.

\begin{lemma}\label{log-mom}
Let \prazess be a \tih \lev process satisfying assumption
$(\mathbb{EM})$ and $f:\Rp\rightarrow\C$ a continuous function such
that $|\Re(f)|\leq M$. Then
\begin{align}
 \E\Bigg[\exp\int_0^t f(s)\ud L_s \Bigg]
 = \exp\int_0^t \theta_s\big(f(s)\big)\ud s.
\end{align}
(The integrals are to be understood componentwise for real and
imaginary part.)
\end{lemma}
\begin{proof}
The proof is similar to the proof of Lemma 3.1 in Eberlein and Raible
\citeyearNP{EberleinRaible99}; see also \citeN[Proposition 1.9]{Kluge05}.
\end{proof}

\section{L\'evy term structure models}
\label{LTSMs}

In this section we review two approaches to modeling the
term structure of interest rates, where the driving process is a
time-inhomogeneous \lev process.

\subsection{The \lev forward rate model}

In the \lev forward rate framework for modeling the term structure
of interest rates, the dynamics of forward rates are specified and
the prices of zero coupon bonds are then deduced. Let $T^*$ be a
fixed time horizon and assume that for every $T\in[0,T^*]$, there
exists a zero coupon bond maturing at $T$ traded in the market; in
addition, let $U\in[0,T^*]$.

The forward rates are driven by a time-inhomogeneous \lev process
$L=\LtT$ on the stochastic basis $(\Omega, \F, \fF,\P)$ with
semimartingale characteristics ($B,C,\nu$) or local characteristics
$(b,c,\lambda)$. The dynamics of the instantaneous continuously
compounded forward rates for $T\in[0,T^*]$ is given by
\begin{align}\label{HJM-fr}
f(t,T) = f(0,T) + \int_0^t \alpha(s,T)\ud s   %\partial_2 A
       - \int_0^t  \sigma(s,T)\ud L_s, \quad 0\leq t\leq T.
\end{align}
The initial values $f(0,T)$ are deterministic, and bounded and
measurable in $T$. In general, $\alpha$ and $\sigma$ are real-valued
stochastic processes defined on $\Omega\times[0,T^*]\times[0,T^*]$
that satisfy the following conditions:
\begin{description}
\item[(A1)] for $s>T$ we have $\alpha(\omega;s,T)=0$ and $\sigma(\omega;s,T)=0$.
\item[(A2)] $(\omega,s,T)\mapsto\alpha(\omega;s,T),\sigma(\omega;s,T)$
            are $\mathcal P\otimes\mathcal B([0,T^*])$-measurable.
\item[(A3)] $S(\omega):=\sup_{s,T\leq T^*}(|\alpha(\omega;s,T)|+|\sigma(\omega;s,T)|)<\infty$.
\end{description}
Then, \eqref{HJM-fr} is well defined and we can find a ``joint''
version of all $f(t,T)$ such that $(\omega;t,T)\mapsto
f(t,T)(\omega)1_{\{t\leq T\}}$ is $\mathcal O\otimes\mathcal
B([0,T^*])$-measurable. Here $\mathcal P$ and $\mathcal O$ denote
the predictable and optional $\sigma$-fields on $\Omega\times[0,T^*]$.

Taking the dynamics of the forward rates as the starting point,
explicit expressions for the dynamics of zero coupon bond prices and
the money market account can be deduced (cf. Proposition 5.2 in
\shortciteNP{BjoerkDimasiKabanovRunngaldier97}). From
\citeN[(2.6)]{EberleinKluge05}, we get that the time-$T$ price of a
zero coupon bond maturing at time $U$ is
\begin{align}\label{HJM-bondprices}
B(T,U) &= \frac{B(0,U)}{B(0,T)}
          \exp\left(\int_0^T\Sigma(s,T,U)\ud L_s -\int_0^T A(s,T,U)\ud s\right),
\end{align}
where the following abbreviations are used:
\begin{align*}
\Sigma(s,T,U) &:= \Sigma(s,U)-\Sigma(s,T),\\
A(s,T,U) &:= A(s,U) - A(s,T),
\end{align*}
and
\begin{align*}
A(s,T) := \int_{s\wedge T}^T \alpha(s,u)\ud u
\quad\mathrm{and}\quad
\Sigma(s,T) := \int_{s\wedge T}^T \sigma(s,u)\ud u.
\end{align*}
Similarly, using \citeN[(2.5)]{EberleinKluge05}, we have for the
money market account
\begin{align}\label{money-market}
B_T^M &= \frac{1}{B(0,T)}
         \exp\left(\int_0^T A(s,T)\ud s - \int_0^T\Sigma(s,T)\ud L_s\right).
\end{align}

In the sequel we will consider only deterministic volatility
structures. Therefore, $\Sigma$ and $A$ are assumed to be
deterministic real-valued functions defined on $\Delta:= \{(s, T)\in
[0,T^*]\times[0,T^*]; s\leq T\}$, whose paths are continuously
differentiable in the second variable. Moreover, they satisfy the
following conditions.
\begin{description}
\item[(B1)] The volatility structure $\Sigma$ is continuous in the first
            argument and bounded in the following way: for
            $(s, T)\in\Delta$  we have
            \begin{align*}
             0\leq \Sigma(s, T)\leq M,
            \end{align*}
            where $M$ is the constant from Assumption ($\mathbb{EM}$).
            Furthermore, we have that $\Sigma(s, T)\neq 0$ for $s<T$ and
            $\Sigma(T,T)=0$ for $T\in[0,T^*]$.
\item[(B2)] The drift coefficients $A(\cdot,T)$ are given by
            \begin{align}\label{HJM-drift}
             A(s, T) = \theta_s(\Sigma(s, T)),
            \end{align}
            where $\theta_s$ is the cumulant generating function associated
            with the triplet $(b_s,c_s,\lambda_s)$, $s\in[0,T]$.
\end{description}

\begin{remark}
The drift condition (\ref{HJM-drift}) guarantees that bond prices
discounted by the money market account are martingales; hence, $\P$
is a martingale measure. In addition, from Theorem 6.4 in
\shortciteN{EberleinJacodRaible05}, we know that the martingale
measure is \emph{unique}.
\end{remark}

\subsection{The \lev forward price model}

In the \lev forward price model the dynamics of forward prices, i.e.
ratios of successive bond prices, are specified. Let
$0=T_0<T_{1}<\cdots<T_{N}<T_{N+1}=T^*$ denote a discrete tenor
structure where $\delta_i=T_{i+1}-T_{i}$, $i\in\{0,1,\dots,N\}$; the
model is constructed via backward induction, hence we denote
by $T_j^*:= T_{N+1-j}$ for $j\in\{0,1,\dots,N+1\}$ and
$\delta^*_j:=\delta_{N+1-j}$ for $j\in\{1,\dots,N+1\}$.

Consider a complete stochastic basis $(\Omega, \F,\fF,\P_{T^*})$
and let $L=\LtT$ be a time-inhomogeneous \lev process satisfying Assumption
$(\mathbb{EM})$. $L$ has semimartingale characteristics
($0,C,\nu^{T^*}$) or local characteristics $(0,c,\lambda^{T^*})$ and
its canonical decomposition is
\begin{align}
L_t = \int_0^t \sqrt{c_s}\ud W_s^{T^*}
     + \int_0^t\int_{\R}x(\mu^L-\nu^{T^*})\dsdx,
\end{align}
where $W^{T^*}$ is a $\P_{T^*}$-standard Brownian motion, $\mu^L$ is
the random measure associated with the jumps of $L$ and $\nu^{T^*}$
is the $\P_{T^*}$-compensator of $\mu^L$. Moreover, we assume that
the following conditions are in force.
\begin{description}
\item[(FP1)] For any maturity $T_{i}$ there exists a bounded, continuous,
             deterministic function $\eta(\cdot,T_{i}):[0,T_i]\rightarrow\R$,
             which represents the volatility of the forward price
             process $F(\cdot, T_{i}, T_{i}+\delta_i)$. Moreover,
             we require that the volatility structure satisfies
             \begin{align*}
             \Big|\sum_{k=1}^i\eta(s,T_{k})\Big|\leq M, &\qquad
             \forall\; i\in\{1,\dots,N\},
             \end{align*}
             for all $s\in[0,T^*]$, where $M$ is the constant from Assumption
             ($\mathbb{EM}$) and $\eta(s,T_i)=0$ for all $s>T_i$.
\item[(FP2)] The initial term structure $B(0,T_i)$, $1\leq i\leq N+1$ is
             strictly positive. Consequently, the initial term structure of
             forward price processes is given, for $1\leq i\leq N$, by
             \begin{align*}
             F(0,T_i,T_i+\delta_i)=\frac{B(0,T_i)}{B(0,T_i+\delta_i)}.
             \end{align*}
\end{description}

The construction starts by postulating that the dynamics of the
forward process with the longest maturity $F(\cdot,T_1^*,T^*)$ are
driven by the time-inhomogeneous \lev process $L$, and evolve as a
martingale under the terminal forward measure $\P_{T^*}$. Then, the
dynamics of the forward processes for the preceding maturities are
constructed by backward induction; therefore, they are driven by the
same process $L$ and evolve as martingales under their associated
forward measures.

Let us denote by $\P_{T^*_{j-1}}$ the forward measure associated
with the settlement date $T^*_{j-1}$, $j\in\{1,\dots,N+1\}$. The
dynamics of the forward price process $F(\cdot,T^*_{j},T^*_{j-1})$
is given by
\begin{align*}
F(t,T^*_{j}\!,T^*_{j-1}) \!=\! F(0,T^*_{j}\!,T^*_{j-1})
        \exp\!\left(\int_0^t b(s,T^*_{j}\!,T^*_{j-1})\ud s\!
                 +\!\!\int_0^t\! \eta(s,T^*_j)\ud L^{T^*_{j-1}}_s\!\right)
\end{align*}
where
\begin{align*}
L^{T^*_{j-1}}_t &=
      \int_0^t\sqrt{c_s}\ud W^{T^*_{j-1}}_s
     +\int_0^t\intrr x(\mu^L-\nu^{T^*_{j-1}})\dsdx
\end{align*}
is a time-inhomogeneous \lev process. Here $W^{T^*_{j-1}}$ is a
$\P_{T^*_{j-1}}$-standard Brownian motion and  $\nu^{T^*_{j-1}}$ is
the $\P_{T^*_{j-1}}$-compensator of $\mu^L$. The forward price
process evolves as a martingale under its corresponding forward
measure, hence, we specify the drift of the forward price process to
be
\begin{align}\label{forwarddrift}
b(s,T^*_{j},T^*_{j-1})
 &=-\half(\eta(s,T^*_j))^2c_s \nonumber\\
 &\quad\;
   -\intrr\left(\e^{\thstj x}-1-\thstj x\right)\lambda_s^{T^*_{j-1}}(\ud x).
\end{align}

The forward measure $\MeaT$, which is defined on
$(\Omega,\F,(\F_t)_{0\leq t\leq T^*_{j-1}})$, is related to the
terminal forward measure $\P_{T^*}$ via
\begin{align*}
\frac{\ud\MeaT}{\ud\P_{T^*}}
 = \prod_{k=1}^{j-1}\frac{F(T^*_{j-1},T^*_k,T_{k-1}^*)}{F(0,T_k^*,T^*_{k-1})}
 = \frac{B(0,T^*)}{B(0,T_{j-1}^*)}\prod_{k=1}^{j-1}F(T^*_{j-1},T^*_k,T_{k-1}^*).
\end{align*}
In addition, the $\P_{T^*_{j-1}}$-Brownian motion is related to the
$\P_{T^*}$-Brownian motion via
\begin{align}\label{ch4:fp-brownians}
W_t^{T^*_{j-1}}
 & = W_t^{T^*_{j-2}} - \int_0^t \eta(s,T^*_{j-1})\sqrt{c_s}\ud s
   = \dots\nonumber\\
 & = W_t^{T^*} - \int_0^t \left(\sum_{k=1}^{j-1}\eta(s,T^*_{k})\right)\sqrt{c_s} \ud s.
\end{align}
Similarly, the $\P_{T^*_{j-1}}$-compensator of $\mu^L$ is related to
the $\P_{T^*}$-compensator of $\mu^L$ via
\begin{align}\label{ch4:fp-nus}
\nu^{T^*_{j-1}} \dsdx
  &= \exp\Big(\eta(s,T^*_{j-1})x\Big)\nu^{T^*_{j-2}}\dsdx
   =\dots \nonumber\\
  &= \exp\left(x\sum_{k=1}^{j-1}\eta(s,T^*_{k})\right)\nu^{T^*}\dsdx.
\end{align}

\begin{remark}\label{ch4:fp-driver}
The process $L=L^{T^*}$, driving the most distant forward price, and
$L^{T^*_{j-1}}$, driving the forward price
$F(\cdot,T^*_{j},T^*_{j-1})$, are both time-inhomogeneous \lev
processes, sharing the same martingale parts and differing only in
the finite variation parts. Applying Girsanov's theorem for
semimartingales yields that the $\P_{T^*_{j-1}}$-finite variation
part of $L$ is
\begin{align*}
\int_0^\cdot c_s\left(\sum_{k=1}^{j-1}\eta(s,T^*_{k})\right) \ds
 + \int_0^\cdot\intrr x\left(\exp\Big(x\sum_{k=1}^{j-1}\eta(s,T^*_{k})\Big)-1\right)\nu^{T^*}\dsdx.
\end{align*}
\end{remark}

\section{Valuation of options on compositions}
\label{compo}

Consider a discrete tenor structure $0=T_0<T_{1}<\cdots<T_{N}<T_{N+1}=T^*$,
where the accrual factor for the time period $[T_{i},T_{i+1}]$ is
$\delta_i=T_{i+1}-T_{i}$, $i\in\{0,1,\dots,N\}$ and let $L(s_i,T_i)$
denote the time-$s_i$ forward LIBOR for the time period $[T_i,T_{i+1}]$.
The \emph{composition} pays a floating rate, typically the LIBOR, compounded
on several consecutive dates. The rates are fixed at the dates $s_i\leq T_i$
and the value of the composition is
\begin{align*}
\prod_{i=1}^{N} \big(1+\delta_i L(s_i, T_i)\big);
\end{align*}
therefore, the composition equals an investment of one currency unit
at the LIBOR rate for $N$ consecutive periods. The value of the
composition is subjected to a cap (or floor) denoted by $K$ and is
settled in arrears, at time $T^*$. Hence, a \emph{cap on the composition}
pays off at maturity the excess of the composition over $K$, i.e.
\begin{align*}%\label{ch4:payoff}
\left(\prod_{i=1}^{N}\big(1+\delta_i L(s_i, T_i)\big)-K\right)^+,
\end{align*}
and similarly, the payoff of a \emph{floor on the composition} is
\begin{align*}
\left(K-\prod_{i=1}^{N}\big(1+\delta_i L(s_i, T_i)\big)\right)^+.
\end{align*}
Notice that without the cap (resp. floor), the payoff of the
composition would simply be that of a floating rate note, where
the proceeds are reinvested. Similarly, if we only consider a
single compounding date, then we are dealing
with a caplet (resp. floorlet), with strike
$\mathscr K:=\frac{K-1}\delta$.

In the following sections, we present methods for the valuation of a
cap on the composition in the L\'evy-driven forward rate and forward
price frameworks. The value of a floor on the composition can either
be deduced via analogous valuation formulae
or via the cap-floor parity for compositions, which reads
\begin{displaymath}
  \mathbf{C}(T^*;K) = \mathbf{F}(T^*;K) + B(0,T_1) - KB(0,T^*).
\end{displaymath}
Here $\mathbf{C}(T^*;K)$ and $\mathbf{F}(T^*;K)$ denote the time-$T_0$
value of a cap, resp. floor, on the composition with cap, resp.
floor, equal to $K$.

\subsection{Forward rate framework}

In this section we derive an explicit formula for the valuation of a
cap on the composition in the \lev forward rate model, making use of
the methods developed in \shortciteN{EberleinGlauPapapantoleon08}. As a special
case, we get valuation formulae for caplets in the \lev forward rate
framework that generalize the results of \citeN{EberleinKluge04},
since we do not require the existence of a Lebesgue density (which is
essential in the convolution representation of option prices; cf.
\citeNP[Chapter 3]{Raible00}).

Firstly, we calculate the quantity that appears in the composition.
By an elementary calculation, we have that
\begin{align*}
\frac{B(s_i,T_i)}{B(s_i,T_{i+1})}
 = \frac{B(0,T_i)}{B(0,T_{i+1})}
   \exp\left( \int_0^{s_i}A(s,T_i,T_{i+1})\ud s
             -\int_0^{s_i}\Sigma(s,T_i,T_{i+1})\ud L_s\right)\!\!.
\end{align*}
Using the fact that $1+\delta_i
L(s_i,T_i)=\frac{B(s_i,T_i)}{B(s_i,T_{i+1})}$ we immediately get
\begin{align*}
\prod_{i=1}^{N}\! \big(1\! + \!\delta_i L(s_i, T_i)\big)\!
 &= \prod_{i=1}^{N}\frac{B(s_i,T_i)}{B(s_i,T_{i+1})}\\
 &= \frac{B(0,T_1)}{B(0,T^*)}\\
 &\times
   \exp\!\left(\sum_{i=1}^{N} \int_0^{s_i}A(s,T_i,T_{i+1})\ud s\!
              -\!\sum_{i=1}^{N} \int_0^{s_i}\Sigma(s,T_i,T_{i+1})\ud L_s\right)\!\!.
\end{align*}
Next, we define the forward measure associated with the date $T^*$
via the Radon--Nikodym derivative
\begin{align*}
\frac{\ud\P_{T^*}}{\ud\P}
 &:= \frac{1}{B^M_{T^*}B(0,T^*)}\\
 & \phantom{:}=
     \exp\left(-\int_0^{T^*} A(s,T^*)\ud s
               +\int_0^{T^*}\Sigma(s,T^*)\ud L_s\right).
\end{align*}
The measures $\P$ and $\P_{T^*}$ are equivalent, since the density
is strictly positive; moreover, we immediately note that
$\E\big[\frac{1}{B^M_{T^*}B(0,T^*)}\big]=1$. The density process
related to the change of measure is given by the restriction of the
Radon--Nikodym derivative to the $\sigma$-field $\F_t$, $t\leq T^*$,
therefore
\begin{align*}
\E\left[ \frac{\ud\P_{T^*}}{\ud\P}\Big|\F_t \right]
 &= \frac{B(t,T^*)}{B^M_{t}B(0,T^*)}\\
 &= \exp\left(-\int_0^{t} A(s,T^*)\ud s
               +\int_0^{t}\Sigma(s,T^*)\ud L_s\right).
\end{align*}
This allows us to determine the tuple of functions that characterize
the process $L$ under this change of measure and we can conclude,
using Theorems III.3.24 and II.4.15 in \citeN{JacodShiryaev03}, that
the driving process $L=(L_t)_{t\in[0,T^*]}$ remains a
time-inhomogeneous \lev process under the measure $\P_{T^*}$.

According to the first fundamental theorem of asset pricing the price
of an option on the composition is equal to its discounted expected
payoff under the martingale measure. Combined with the forward
measure defined above, this gives
\begin{align*}
\mathbf{C}(T^*;K)
        &= \E_{\P}\left[\frac{1}{B^M_{T^*}}
           \left(\prod_{i=1}^{N}\frac{B(s_i,T_i)}{B(s_i,T_{i+1})}-K\right)^+\right]\\
        &= B(0,T^*) \E_{\P_{T^*}}
           \left[\left(\prod_{i=1}^{N}\frac{B(s_i,T_i)}{B(s_i,T_{i+1})}-K\right)^+\right]\\
        &= B(0,T^*) \E_{\P_{T^*}}\left[\left(\exp H -K\right)^+\right],
\end{align*}
where the random variable $H$ is defined as
\begin{align*}%\label{ch4:HJM-compo-H}
 H := \log\frac{B(0,T_1)}{B(0,T^*)}
    + \sum_{i=1}^{N} \int_0^{s_i}A(s,T_i,T_{i+1})\ud s
    - \sum_{i=1}^{N} \int_0^{s_i}\Sigma(s,T_i,T_{i+1})\ud L_s.
\end{align*}

Let us denote by $M^{T^*}_{H}$ the moment generating function of $H$
under the measure $\P_{T^*}$. The next theorem provides an
analytical expression for the value of a cap on the composition.
Preceding that, we provide an expression for $M^{T^*}_{H}(z)$ for
suitable complex arguments $z$.

\begin{lemma}\label{comp-HJM-mgf}
Let $M$  and $\varepsilon$ be suitably chosen such that
$\Sigma(s,T)\leq M'$ for all $s,T\in[0,T^*]$ and
$\Sigma(s,T_{i+1})1_{[s_i,s_{i+1}]}(s)\leq M''$ for all
$s,s_i,T_{i+1}\in[0,T^*]$, where $0<M''<M'<M$ and
$\frac{M}{M''}>N+1$. Then, for each
$R\in\mathcal{I}_2=[1-\frac{M-M''(N+1)}{M'+M''(N+1)},1+\frac{M-M''(N+1)}{M'+M''(N+1)}]$,
we have that $M^{T^*}_H(R)<\infty$ and for every $z\in\C$ with $\Re
z=R$
\begin{align*}
M^{T^*}_H(z) &= \mathcal Z^z
 \exp \left( \int_0^{T^*}\left( z \sum^{N}_{i=1} A(s,T_i,T_{i+1})1_{[0,s_i]}(s)
             - \theta_s\Big(\Sigma(s,T^*)\Big)
 \right.\right.\nonumber\\
 &\qquad\qquad   \left.\phantom{\int_0^{T^*}}\left.
                 + \theta_s\Big(\Sigma(s,T^*)-z\sum^{N}_{i=1}\Sigma(s,T_i,T_{i+1})1_{[0,s_i]}(s)\Big)
  \right)\ds\right)
\end{align*}
where $\mathcal Z:=\frac{B(0,T_1)}{B(0,T^*)}$.
\end{lemma}
\begin{proof}
Fix an $R\in\mathcal{I}_2$. Then, for
$z\in\C$ with $\Re z=R$, and denoting by
$\underline{\overline{\Sigma}}(s,T)=\sum_{i=0}^{N}\Sigma(s,T_{i+1})1_{[s_i,s_{i+1}]}(s)$,
we get that
\begin{eqnarray}\label{ch4:M-bound-I}
\lefteqn{\Big|\Re\Big(-z\sum^{N}_{i=1}\Sigma(s,T_i,T_{i+1})1_{[0,s_i]}(s)\Big)+\Sigma(s,T^*)\Big|}
  \nonumber\\
 &=& \Big|\Re\Big(z\sum_{i=0}^{N}\Sigma(s,T_{i+1})1_{[s_i,s_{i+1}]}(s)-z\Sigma(s,T^*)\Big)
           +\Sigma(s,T^*)\Big|\nonumber\\
 &=& \Big|\Re\Big((1-z)\big(\Sigma(s,T^*)-\underline{\overline{\Sigma}}(s,T)\big)\Big)
          +\underline{\overline{\Sigma}}(s,T)\Big| \nonumber\\
 &\leq& |1-R||\Sigma(s,T^*)-\underline{\overline{\Sigma}}(s,T)|
          + |\underline{\overline{\Sigma}}(s,T)|   \nonumber\\
 &\leq& \frac{M-M''(N+1)}{M'+M''(N+1)}(M'+M''(N+1)) + M''(N+1)
  =M.
\end{eqnarray}
Now, define the constants
\begin{align*}
\mathcal Z_0:=\exp\left(z\bigg(\log\frac{B(0,T_1)}{B(0,T^*)}
  + \int_0^{T^*}\sum^{N}_{i=1}A(s,T_i,T_{i+1})1_{[0,s_i]}(s)\ds\bigg)\right)
\end{align*}
and
 $\mathcal Z_1:=\mathcal Z_0\times\exp\left(-\int_0^{T^*}A(s,T^*)\ds\right)$.
Hence, the moment generating function of $H$ is
\begin{align*}
M^{T^*}_H(z)
 &= \E_{\P_{T^*}}\Big[\exp(zH)\Big]\\
 &= \E_{\P_{T^*}}\left[\exp\left(z\bigg(\log\frac{B(0,T_1)}{B(0,T^*)}
                         + \sum_{i=1}^{N} \int_0^{s_i}A(s,T_i,T_{i+1})\ud s\right.\right.\\
 & \qquad\qquad\qquad\qquad\qquad\left.\left.
                         - \sum_{i=1}^{N} \int_0^{s_i}\Sigma(s,T_i,T_{i+1})\ud L_s\bigg)\right)\right]\\
 &= \exp\left(-\int_0^{T^*}A(s,T^*)\ds\right) \times \mathcal Z_0\\
 & \qquad\times
    \E_{\P}\left[\exp\left(-z\sum_{i=1}^{N} \int_0^{s_i}\Sigma(s,T_i,T_{i+1})\ud L_s
                           +\int_0^{T^*}\Sigma(s,T^*)\ud L_s\right)\right]\\
 &= \mathcal Z_1 \E_{\P}\left[\exp\int_0^{T^*}\left(-z\sum_{i=1}^{N} \Sigma(s,T_i,T_{i+1})1_{[0,s_i]}(s)
                                          +\Sigma(s,T^*)\right)\ud L_s\right]\\
 &= \mathcal Z_1 \exp\int_0^{T^*}\left(
                  \theta_s\Big(-z\sum^{N}_{i=1}\Sigma(s,T_i,T_{i+1})1_{[0,s_i]}(s)
                           +\Sigma(s,T^*)\Big) \right)\ds,
\end{align*}
where for the last equality we have applied Lemma \ref{log-mom},
which is justified by \eqref{ch4:M-bound-I}. In addition, we get
that $M^{T^*}_H(R)<\infty$ for $R\in\mathcal I_2$.
\end{proof}

\begin{theorem}
Assume that forward rates are modeled according to the \lev forward
rate model. The price of a cap on the composition is
\begin{align*}
\mathbf{C}(T^*;K)
 = \frac{B(0,T^*)}{2\pi}
   \int_{\R} M^{T^*}_{H}(R-iu) \frac{K^{1+iu-R}}{(iu-R)(1+iu-R)}\ud u,
\end{align*}
where $M^{T^*}_{H}$ is given by Lemma \ref{comp-HJM-mgf} and
$R\in(1,1+\frac{M-M''(N+1)}{M'+M''(N+1)}]$. %I_1\cap\mathcal{I}_2
\end{theorem}
\begin{proof}
Firstly, let us recall that the Fourier transform of the payoff function
$f(x)=(\e^x-K)^+$, $K\in\Rp$, corresponding to a call option is
\begin{align}\label{call-fourier}
 \widehat{f}(z)= \frac{K^{1+iz}}{iz(1+iz)},
\end{align}
for $z\in\C$ with $\Im z\in(1,\infty)=:\mathcal I_1$; cf. Example 3.15
in \citeN{Papapantoleon06}.

Now, since the prerequisites of Theorem 2.2 in
\shortciteN{EberleinGlauPapapantoleon08} are satisfied for
$R\in \mathcal I_1\cap\mathcal{I}_2$, we immediately have that
\begin{align*}
\mathbf{C}(T^*;K)
  &= B(0,T^*) \E_{\P_{T^*}}\left[\left(\e^H -K\right)^+\right]\\
  &= \frac{B(0,T^*)}{2\pi} \int_{\R} M^{T^*}_{H}(R-iu) \frac{K^{1+iu-R}}{(iu-R)(1+iu-R)}\ud u,
\end{align*}
and the assertion is proved.
\end{proof}

\subsection{Forward price framework}

The aim of this section is to derive an explicit formula for the
valuation of a cap on the composition in the \lev forward price
model. Once again, the valuation formulae will be based on the
methods developed in \shortciteN{EberleinGlauPapapantoleon08}.

We begin by noticing that the quantity that appears in the
composition can be expressed in terms of forward prices, since
\[
1+\delta_i L(\cdot,T_i) = F(\cdot,T_i,T_{i+1}),
\]
and the forward prices are the modeling object in this framework. We
know that each forward price process evolves as a martingale under
its corresponding forward measure; moreover, we know that all
forward price processes are driven by the same time-inhomogeneous
\lev process (see also Remark \ref{ch4:fp-driver}). Therefore, we
will carry out the following program to arrive at the valuation
formulae:
\begin{enumerate}
\item lift all forward price processes from their forward measure to
      the terminal forward measure;
\item calculate the product of the composition factors;
\item price the cap on the composition as a call option on this product.
\end{enumerate}

Appealing to the structure of the forward price process and the
connection between the Brownian motions and the compensators under
the different measures, cf. equations \eqref{ch4:fp-brownians} and
\eqref{ch4:fp-nus}, we get that
\begin{align}\label{ch4:prod1}
F(t,T^*_{j}\!,T^*_{j-1})
 &= F(0,T^*_{j}\!,T^*_{j-1})
        \exp\!\left(\int_0^t b(s,T^*_{j}\!,T^*_{j-1})\ud s
                     +\!\int_0^t\! \eta(s,T^*_j)\ud L^{T^*_{j-1}}_s\!\right)\nonumber\\
 &= F(0,T^*_{j}\!,T^*_{j-1})
        \exp\!\left(\int_0^t b(s,T^*_{j}\!,T^*)\ud s
                     +\!\int_0^t\! \eta(s,T^*_j)\ud L^{T^*}_s\!\right),
\end{align}
for all $j\in{\{1,\dots,N+1\}}$. Here $L^{T^*}$ is the driving \tih \lev process with
$\P_{T^*}$-canonical decomposition
\begin{align}\label{ch4:fp-driverII}
L^{T^*}_t = \int_0^t \sqrt{c_s}\ud W_s^{T^*}
           + \int_0^t\int_{\R}x(\mu^L-\nu^{T^*})\dsdx,
\end{align}
and the drift term of the forward process
$F(\cdot,T^*_{j},T^*_{j-1})$ under the terminal measure $\P_{T^*}$,
is
\begin{align}\label{ch4:prod2}
b(s,T^*_{j}\!,T^*) &= -c_s \left(\frac12(\eta(s,T_j^*))^2
                         + \eta(s,T_j^*)\sum_{k=1}^{j-1}\eta(s,T_k^*)\right)\nonumber\\
                   &\quad
     -\int_\R\left(\big(\e^{x\eta(s,T_j^*)}-1\big)\e^{x\sum_{k=1}^{j-1}\eta(s,T_k^*)}
                     -x\eta(s,T_j^*)\right) \lambda_s^{T^*}(\dx).
\end{align}
It is immediately obvious from \eqref{ch4:prod1},
\eqref{ch4:fp-driverII} and \eqref{ch4:prod2} that
$F(\cdot,T^*_{j},T^*_{j-1})$ is not a $\P_{T^*}$-martingale, unless
$j=1$ (where we use the convention that $\sum_{j=1}^0=0$).

Now, the composition takes the following form
\begin{align}\label{ch4:compo-fp-prod}
\prod_{i=1}^{N}\big(1+\delta_i L(s_i, T_i)\big)
 &= \prod_{j=1}^{N}F(s^*_j,T^*_{j},T^*_{j-1})\nonumber\\
 &= \frac{B(0,T_N^*)}{B(0,T^*)}\\\nonumber
 &\;\;\times
  \exp\left( \sum_{j=1}^{N} \int_0^{s_j^*}b(s,T_j^*,T^*)\ud s
              +\sum_{j=1}^{N} \int_0^{s_j^*}\eta(s,T_j^*)\ud L^{T^*}_s\right)\!\!.
\end{align}
where $s_j^*=s_{N+1-j}$, $j\in\{1,\cdots,N\}$. Define the random variable
\begin{align}\label{ch4:compo-fp-H}
H:= \log \frac{B(0,T_N^*)}{B(0,T^*)}
  + \sum_{j=1}^{N} \int_0^{s_j^*}b(s,T_j^*,T^*)\ud s
  + \sum_{j=1}^{N} \int_0^{s_j^*}\eta(s,T_j^*)\ud L^{T^*}_s
\end{align}
and now we can express the option on the composition as an option
depending on this random variable. The next theorem provides a
formula for the valuation of a cap on the composition.

\begin{theorem}
Let forward prices be modeled according to the \lev forward
process framework. Then, the price of a cap on the composition is
\begin{align}\label{ch4:fp-price-compo}
 \mathbf{C}(T^*;K)
  = \frac{B(0,T^*)}{2\pi}
    \int_\R M_H(R-iu)\frac{K^{1+iu-R}}{(iu-R)(1+iu-R)}\ud u,
\end{align}
where the moment generating function of $H$ is given by Lemma
\ref{ch4:compo-fp-mgf} and $R\in(1,\frac{M}{M'}]$.
\end{theorem}
\begin{proof}
The option on the composition is priced under the terminal forward
martingale measure $\P_{T^*}$. Using \eqref{ch4:compo-fp-prod} and
\eqref{ch4:compo-fp-H}, we can express the cap on the composition as
a call option depending on the random variable $H$. Then we get
\begin{align*}
\mathbf{C}(T^*;K)
 &= B(0,T^*)\E_{\P_{T^*}}\left[\left(\prod_{j=1}^{N}F(s^*_j,T^*_{j},T^*_{j-1})-K\right)^+\right]\\
 &= B(0,T^*)\E_{\P_{T^*}}\left[\left(\e^H-K\right)^+\right]\\
 & = \frac{B(0,T^*)}{2\pi}
    \int_\R M_H(R-iu)\frac{K^{1+iu-R}}{(iu-R)(1+iu-R)}\ud u
\end{align*}
where we have applied Theorem 2.2 in \shortciteN{EberleinGlauPapapantoleon08}
and used \eqref{call-fourier} once again.
\end{proof}

\begin{lemma}\label{ch4:compo-fp-mgf}
Let $M$ and $\varepsilon$ be suitably chosen such that
$|\sum_{k=1}^N\eta(s,T_k)|\leq M'$ for some $M'<M$ and for all
$s\in[0,T^*]$. Then, for each $R\in[0,\frac{M}{M'}]$ we have that
$M_H(R)<\infty$, and for every $z\in\C$ with $\Re
z\in[0,\frac{M}{M'}]$ the moment generating function of $H$ is
\begin{align*}
M_H(z)
 &= \mathcal Z^z
     \exp\left(\int_{0}^{T^*}
           \bigg(z\sum_{j=1}^{N}b(s,T^*_j,T^*)
                 +\theta_s^{T^*}\Big(z\sum_{j=1}^{N}\eta(s,T^*_j)\Big)
           \bigg)\ds\right),
\end{align*}
where $\mathcal Z=\frac{B(0,T^*_N)}{B(0,T^*)}$ and $\theta_s^{T^*}$
is the cumulant generating function associated with the triplet
$(0,c_s,\lambda_s^{T^*})$.
\end{lemma}
\begin{proof}
Fix an $R\in[0,\frac{M}{M'}]$ and then for $z\in\C$ with $\Re z=R$
we get
\begin{align}\label{ch4:M-bound-II}
\left|\Re\left(z\sum_{k=1}^N\eta(s,T_k)\right)\right|
  =   R\left|\sum_{k=1}^N\eta(s,T_k)\right|
  \leq \frac{M}{M'}M' = M.
\end{align}

Now, define the constant
\begin{align*}
\mathcal Z_2:=\left(\frac{B(0,T^*_N)}{B(0,T^*_{1})}\right)^z
          \exp\left(z\sum_{j=1}^{N}\int_{0}^{s^*_j}b(s,T^*_j,T^*)\ds\right).
\end{align*}
Then, similarly to the proof of Lemma \ref{comp-HJM-mgf}, we get
\begin{align*}
M_H(z)
  &= \mathcal Z_2 \E_{\P_{T^*}}\left[
        \exp\left(\int_0^{T^*}z\sum_{j=1}^{N} \eta(s,T_j^*)1_{[0,s^*_j]}(s)\ud L^{T^*}_s\right)\right]\\
  &= \mathcal Z_2
        \exp\left(\int_{0}^{T^*}\theta_s^{T^*}\Big(z\sum_{j=1}^{N}\eta(s,T^*_j)\Big)\ds\right),
\end{align*}
where for the last equality we have applied Lemma \ref{log-mom},
which is justified by \eqref{ch4:M-bound-II}. Note also that
$\eta(s,T_j^*)=0$ for $s>s^*_j$, which is the fixing date for the
rate; accordingly, $b(s,T^*_j,T^*)=0$ for $s>s^*_j$, cf.
\eqref{ch4:prod2}.

In addition, we get that $M_H(R)<\infty$ for
$R\in[0,\frac{M}{M'}]$.
\end{proof}

\subsection{Numerical illustration}

In order to get an idea about the difference in the prices of an option
on the composition of LIBOR rates between the classical, Brownian-driven,
HJM model and the L\'evy forward rate model, we set up an artificial but
reasonable market. Our aim here is not to give a complete analysis but
rather a flavor of the impact we can expect. As an example, we look at a
5 year caplet on a composition of LIBOR rates.

Rates are assumed to be flat at 4\% across all maturities and instruments.
Moreover, we assume that all volatilities are marked according to the SABR
model (cf. \shortciteNP{HaganKumarLesniewskiWoodward02}) with parameters
$\sigma=1\%$, $\alpha = 40\%$, $\beta=0\%$ and $\rho=30\%$.
For the calibration, we choose the Vasi\v{c}ek volatility structure with a
fixed parameter $a=0.05$. Not surprisingly, the normal inverse Gaussian
L\'evy model calibrates
well to the market smile whereas the classical HJM model only gets the ATM
point right but cannot reproduce any smile or skew; see Figure \ref{Fig1}.

\begin{figure}
\begin{center}
 \includegraphics[height=6.0cm,keepaspectratio=true]{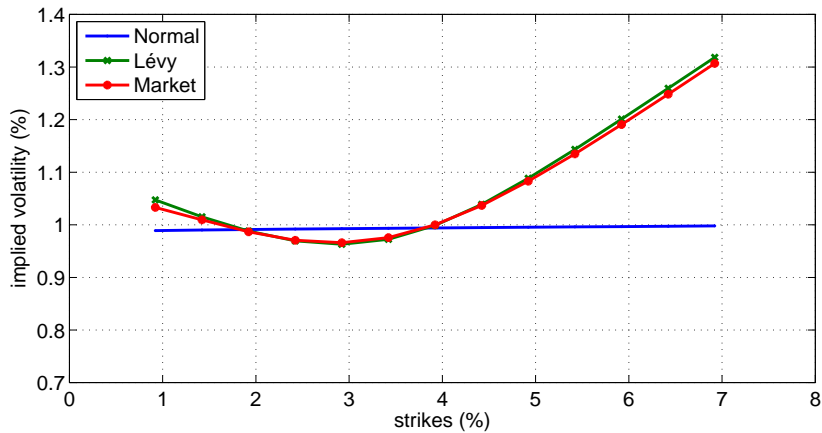}
  \caption{Implied normal volatilities for a 5Y option on the 3M LIBOR.}
  \label{Fig1}
\end{center}
\end{figure}

The results for a 5Y caplet on a composition of LIBOR rates are shown in
Figure \ref{Fig2}. It should not be surprising that the classical model
does not produce any smile whereas the L\'evy model does; note that the
ATM prices are also different, where ATM $\approx1.22$. Moreover, as has
been observed in many
other situations, the Brownian-driven model overprices the ATM options
and underprices the in- and out-of-the-money options compared to the more
realistic \lev model.

\begin{figure}
\begin{center}
 \includegraphics[height=6.0cm,keepaspectratio=true]{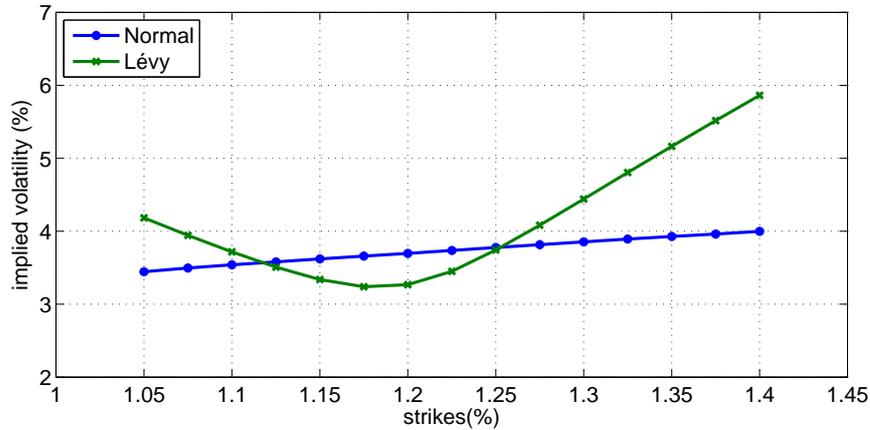}
  \caption{Implied normal volatilities for a 5Y option on composition.}
  \label{Fig2}
\end{center}
\end{figure}

\section{Conclusion}
\label{conclusion}

We have presented valuation formulas, based on Fourier transforms, for
pricing an option on the composition of LIBOR rates in the forward rate
and forward price models driven by \tih \lev processes. Analogous formulas
can also be derived for the affine `forward price'-type framework proposed
by \citeN{KellerResselPapapantoleonTeichmann09}; this framework combines
the analytical tractability of the forward price framework described here
with \emph{positive} LIBOR rates.

The challenge ahead is to derive valuation formulas in the LIBOR model
driven by a \tih \lev process. This task requires some sophisticated
approximations due to the structure of the dynamics in LIBOR market models;
the interested reader is referred to \citeANP{SiopachaTeichmann07}
\citeyear{SiopachaTeichmann07} and
\citeN{HubalekPapapantoleonSiopacha09} for a detailed analysis.

\bibliographystyle{chicago}
%\bibliography{/home/famuser/papapan/Papers/references}
\bibliography{C:/WorkFiles/Papers/references}

\end{document}